\documentclass{article}

\usepackage{amssymb, amsmath, amsthm,authblk}
\usepackage{natbib}
\usepackage{lmodern}
\usepackage[utf8]{inputenc}
\usepackage{amsmath}
\usepackage{amsfonts}
\usepackage{amssymb}
\usepackage{amsmath}
\usepackage{amsthm}
\usepackage[a4paper,left=3.3cm,right=3.3cm,bottom=2.2cm,top=2cm,foot=1.3cm,includeheadfoot=true]{geometry}
\usepackage{framed}
\usepackage[boxed]{algorithm2e}
\usepackage[hidelinks]{hyperref}
\usepackage{multirow,bigdelim}
\usepackage[babel=true]{csquotes}
\usepackage{tikz}
\usepackage{url}

\newtheorem{lemma}{Lemma}
\newtheorem{proposition}[lemma]{Proposition}

\newtheorem{theorem}[lemma]{Theorem}
\newtheorem{corollary}[lemma]{Corollary}

\usepackage{booktabs} 
\usepackage{tikz}
\usepackage{mathdots}
\usepackage{esvect}
\usepackage{bbm}

\usepackage{xcolor}
\definecolor{darkgreen}{rgb}{0,.35,0}
\definecolor{darkblue}{rgb}{0,0,.5}
\definecolor{darkred}{rgb}{.6,0,0}

\usepackage{hyperref}
\hypersetup{pdfauthor={Mark Giesbrecht, Armin Jamshidpey, \'Eric Schost},
  pdftitle={Quadratic Probabilistic Algorithms for Normal Bases},
  bookmarksnumbered=true,colorlinks,linkcolor=darkblue,
 citecolor=darkgreen, urlcolor=darkred}

\numberwithin{equation}{section}

\newcommand{\osumcost}{O(n^{(3/4)\cdot \omega(4/3)})}
\newcommand{\osumcosttilde}{\tilde{O}(n^{(3/4)\cdot \omega(4/3)})}
\newcommand{\thecost}{\tilde{O}(\vert G \vert ^{(3/4)\cdot \omega(4/3)})}

\newcommand{\FF}{{\mathbb{F}}}
\newcommand{\xbar}{\xi}
\newcommand{\zbar}{\zeta}
\newcommand{\alg}{quadratic\,}
\newcommand{\citeN}{\cite}

{
      \theoremstyle{acmplain}
      \newtheorem{assumption}{Assumption}
  }

{
      \theoremstyle{acmplain}
      
  }

\title{Quadratic Probabilistic Algorithms for Normal Bases}

\author{Mark Giesbrecht\footnote{mwg@uwaterloo.ca}}
\author{Armin Jamshidpey \footnote{ armin.jamshidpey@uwaterloo.ca}}
\author{\'Eric Schost \footnote{eschost@uwaterloo.ca}}
\affil{David Cheriton School of Computer Science\\ University of Waterloo, Canada}

\begin{document}
\maketitle

\newcommand{\F}{{\mathsf{F}}}
\newcommand{\K}{{\mathsf{K}}}

\newcommand{\NN}{{\mathbb{N}}}
\newcommand{\N}{{\mathbb{N}}}

\def\A{\mathbb{A}}
\def\H{\mathbb{H}}
\def\B{\mathbb{B}}
\def\Z{\mathbb{Z}}
\def\C{\mathbb{C}}
\def\Q{\mathbb{Q}}
\def\D{\mathbb{D}}
\newcommand{\QQ}{\mathbb{Q}}
\newcommand{\mat}[1]{\mathbf{\MakeUppercase{#1}}} 

\begin{abstract}
  It is well known that for any finite Galois extension field $\K/\F$,
  with Galois group $G = \mathrm{Gal}(\K/\F)$, there exists an element
  $\alpha \in \K$ whose orbit $G\cdot\alpha$ forms an $\F$-basis of
  $\K$. Such an element $\alpha$ is called \emph{normal} and
  $G\cdot\alpha$ is called a normal basis. In this paper we introduce
  a probabilistic algorithm for finding a normal element when $G$ is
  either a finite abelian or a metacyclic group. The algorithm is
  based on the fact that deciding whether a random element $\alpha \in
  \K$ is normal can be reduced to deciding whether $\sum_{\sigma \in
    G} \sigma(\alpha)\sigma \in \K[G]$ is invertible. In an algebraic
  model, the cost of our algorithm is quadratic in the size of $G$ for
  metacyclic $G$ and slightly subquadratic for abelian $G$.
\end{abstract}

\maketitle

\section{Introduction}

For a finite Galois extension field $\K/\F$, with Galois group $G =
\mathrm{Gal}(\K/\F)$, an element $\alpha \in \K$ is called
\emph{normal} if the set of its Galois conjugates $G \cdot \alpha =
\{ \sigma(\alpha): \sigma\in G\}$ forms a basis for $\K$ as a vector space over
$\F$. The existence of normal element for any finite Galois extension
is classical, and constructive proofs are provided in most algebra texts
(see, e.g., \cite{Lang}, Section 6.13).
 
While there is a wide range of well-known applications of normal bases in
finite fields, such as fast exponentiation~\cite{GaGaPaSh00}, there also
exist applications of normal elements in characteristic zero.  For instance,
in multiplicative invariant theory, for a given permutation lattice and
related Galois extension, a normal basis is useful in computing the
multiplicative invariants explicitly~\cite{Jam18}.

A number of algorithms are available for finding a normal element in
characteristic zero fields and finite fields.  Because of their immediate
applications in finite fields, algorithms for determining normal elements
in this case are most commonly seen.  A fast randomized algorithm for
determining a normal element in a finite field $\FF_{q^n}/\FF_q$, where
$\FF_{q^n}$ is the finite field with $q^n$ elements for any prime power $q$
and integer $n>1$, is presented by \citeN{GatGie90}, with a cost of
$O(n^2+n\log q)$ operations in $\FF_q$.  A faster randomized algorithm is
introduced by \citeN{KalSho98}, with a cost of $O(n^{1.815}\log q)$
operations in $\FF_q$.  In the bit complexity model, Kedlaya and Umans showed
how to reduce the exponent of $n$ to $1.5+\varepsilon$ (for any
$\varepsilon > 0$), by leveraging their quasi-linear time algorithm for
{\em modular composition}~\cite{KeUm11}. \citeN{LenstraNormal} introduced a
deterministic algorithm to construct a normal element which uses $n^{O(1)}$
operations in $\FF_{q^n}/\FF_q$.  To the best of our knowledge, the
algorithm of \citeN{AugCam94} is the most efficient deterministic method,
with a cost of $O(n^3+n^2\log q)$ operations in $\FF_q$.

In characteristic zero, \citeN{SchSte93} gave an algorithm for finding
a normal basis of a number field over $\QQ$ with a cyclic Galois group
of cardinality $n$ which requires $n^{O(1)}$ operations in $\QQ$.
\citeN{Pol94} gives an algorithm for the more general case of finding
a normal basis in an abelian extension $\K/\F$ which requires
$n^{O(1)}$ in $\F$.  More generally in characteristic zero, for any
Galois extension $\K/\F$ of degree $n$ with Galois group given by a
collection of $n$ matrices, \citeN{Girstmair} gives an algorithm which
requires $O(n^4)$ operations in $\F$ to construct a normal element in
$\K$.

In this paper we present a new randomized algorithm for finding a normal
element for abelian and metacyclic extensions, with a runtime \alg
in the degree $n$ of the extension. The costs of all algorithms are
measured by counting \emph{arithmetic operations} in $\F$ at unit cost.
Questions related to the bit-complexity of our algorithms are challenging,
and beyond the scope of this paper.

Our main conventions are the following.
\begin{assumption}
  \label{assum}
  Let $\K/\F$ be a finite Galois extension presented as
  $\K=\F[x]/\langle P(x)\rangle$, for an irreducible polynomial $P\in
  \F[x]$ of degree $n$, with $\F$ of characteristic zero. Then,
  \begin{itemize}
  \item elements of $\K$ are written on the power basis $1,\xbar,\dots,\xbar^{n-1}$,
    where $\xbar := x \bmod P$;
  \item elements of $G$ are represented by their action on $\xbar$.
  \end{itemize}
\end{assumption}

In particular, for $g \in G$ given by means of $\gamma:=g(\xbar) \in \K$,
and $\beta = \sum_{0\leq i<n}\beta_i\xbar^i\in\K$, the fact that $g$ is an
$\F$-automorphism implies that $g(\beta)$ is equal to $\beta(\gamma)$, the
polynomial composition of $\beta$ at $\gamma$ (reduced modulo $P$).

Our algorithms combine techniques and ideas due
to~\cite{GatGie90,KalSho98}: $\alpha \in \K$ is normal if and only if
the element $S_\alpha := \sum_{g \in G} g(\alpha)g \in \K[G]$ is
invertible in the group algebra $\K[G]$. The algorithms choose
$\alpha$ at random; a generic choice is normal (so we expect $O(1)$
random trials to be sufficient). However, writing down $S_\alpha$
involves $\Theta(n^2)$ elements in $\F$, which precludes a
subquadratic runtime. Instead, knowing $\alpha$, the algorithms use a
randomized reduction to a similar question in $\F[G]$, that amounts to
applying a random projection $\ell:\K\to\F$ to all entries of
$S_\alpha$, giving us an element $s_{\alpha,\ell} \in \F[G]$. For
that, we adapt algorithms from~\cite{KalSho98} that were written for
Galois groups of finite fields.

Having $s_{\alpha,\ell}$ in hand, we need to test its
invertibility. In order to do so, we present an algorithm in the
abelian case which relies on the fact that $\F[G]$ is isomorphic to a
multivariate quotient polynomial ring by an ideal $(x^{e_i}_i-1)_{1
  \leq i \leq m}$, where $e_i$'s are positive integers. 

For metacyclic groups, two algorithms are introduced to solve the same
problem; which one is faster depends on the parameters defining our
group. Both algorithms are based on testing the invertibility of an
injective homomorphic image of $s_{\alpha,\ell}$ in a matrix algebra
over a product of fields. These questions are closely related to that
of Fourier transform over $G$, and it is worth mentioning that there
is a vast literature on fast algorithms for Fourier transforms (over
the base field $\C$). Relevant to our current context, we invite the
reader to consult \cite{ClaMu04,MaRockWol18} and references therein
for details. At this stage, it is not clear how we can apply these
methods in our context.

This paper is written from the point of view of obtaining improved
asymptotic complexity estimates. Since our main goal is to highlight
the exponent (in $n$) in our runtime analyses, costs are given using
the soft-O notation: $S(n)$ is in $\tilde{O}(T(n))$ if it is in
$O(T(n) \log(T(n))^c)$, for some constant $c$.

The main result of this paper is the following theorem; we use a
constant $\omega(4/3)$ that describes the cost of certain rectangular
matrix products (see the end of this section).
\begin{theorem}\label{thm:main}
  Under Assumption \ref{assum}, a normal element of $\K$ can be found
  using $\thecost$ operations in $\F$ if $G$ is abelian, with
  $(3/4)\cdot\omega(4/3)<1.99$.  Moreover, the same problem for
  metacyclic groups can be solved using $\tilde{O}(\vert G \vert^2)$
  operations in $\F$. The algorithms are randomized.
\end{theorem}
Although the cost is quadratic in the size of input for a general
metacyclic group, in many cases it will be (slightly) subquadratic,
under specific conditions on the parameters defining $G$ (see
Section~\ref{sec:invertibility}).

Section \ref{sec:pre} of this paper is devoted to definitions and
preliminary discussions.  In Section \ref{sec:osum}, two
subquadratic-time algorithms are presented for the randomized reduction
of our main question to invertibility testing in $\F[G]$, for
respectively abelian and metacyclic groups.  Finally, in Section
\ref{sec:invertibility}, we show that the latter problem can be solved
in quasi-linear time for an abelian group; for metacyclic
groups, we give a quadratic time algorithm, and discuss cases when
the cost can be improved.

Our algorithms make
extensive of known algorithms for polynomial and matrix arithmetic; in
particular, we use repeatedly the fact that polynomials of degree $d$ in
$\F[x]$, for any field $\F$ of characteristic zero, can be multiplied in
$\tilde{O}(n)$ operations in $\F$~\cite{ScSt71}. As a result, arithmetic 
operations $(+,\times,\div)$ in $\K$ can all be done using $\tilde{O}(n)$ 
operations in $\F$~\cite{vzGathen13}.

For matrix arithmetic, we will rely on some non-trivial results on
rectangular matrix multiplication initiated by \citeN{LoRo83}. For $k \in
\mathbb{R}$, we denote by $\omega(k)$ a constant such that over any
ring, matrices of sizes $(n,n)$ by $(n,\lceil n^k \rceil)$ can be
multiplied in $O(n^{\omega(k)})$ ring operations (so $\omega(1)$ is
the usual exponent of square matrix multiplication, which we simply
write $\omega$).  The sharpest values known to date for most
rectangular formats are from~\cite{LeGall}; for $k=1$, the best known
value is $\omega \le 2.373$ by \citeN{LeGall14}. Over a field, we will
frequently use the fact that further matrix operations (determinant or
inverse) can be done in $O(n^\omega)$ base field operations.


\section{Preliminaries}
\label{sec:pre}

One of the well-known proofs of the existence of a normal element for a
finite Galois extension \cite[Theorem 6.13.1]{Lang} suggests a randomized
algorithm for finding such an element. Assume $\K/\F$ is a finite Galois
extension with Galois group $G = \lbrace g_1 , \ldots , g_n \rbrace$. If
$\alpha \in \K$ is a normal element, then
\begin{equation}
  \label{eq:fstrow}
  \sum_{j=1}^n 
  c_j g_j(\alpha)=0, \,\,\, c_j \in \F 
\end{equation} 
implies $c_1 =\dots=c_n = 0$. For each
$i \in \lbrace 1, \ldots , n\rbrace$, applying $g_i$ to equation
\eqref{eq:fstrow} yields
\begin{equation} \label{eq:otherrow} \sum_{j=1}^n c_j g_i g_j(\alpha)=0.
\end{equation}
Using \eqref{eq:fstrow} and \eqref{eq:otherrow}, one can form the linear
system $\mat{M}_G(\alpha)\boldsymbol{c} = \textbf{0}$,
with $\boldsymbol{c} = [ c_1~ \cdots~c_n]^T$ and
 where, for $\alpha\in\K$,
\[
  \mat M_G(\alpha) =
  \begin{bmatrix}
    g_1 g_1(\alpha) & g_1 g_2(\alpha) & \cdots & g_1 g_n(\alpha) \\
    g_2 g_1(\alpha) & g_2 g_2(\alpha) & \cdots & g_2 g_n(\alpha) \\
    \vdots		& \vdots	& \vdots & \vdots \\
    g_n g_1(\alpha) & g_n g_2(\alpha) & \cdots & g_n g_n(\alpha) \\
  \end{bmatrix} \in M_n(\K).
\]
Classical proofs then proceed to show that there exists $\alpha \in \K$
with $\det(\mat M_G(\alpha))\neq 0$.
 
This approach can be used as the basis of a randomized algorithm for
finding a normal element: choose a random element $\alpha$ in $\K$ until we
find one such that $ \mat M_G(\alpha)$ is invertible. A direct implementation computes
all the entries of the matrix and then uses linear algebra to compute its
determinant; using fast matrix arithmetic this requires $O(n^\omega)$
operations in $\K$, that is $\tilde{O}(n^{\omega+1})$ operations in
$\F$. This is at least cubic in $n$, and only a minor improvement over the
previously best-known approach of \citeN{Girstmair}. The main contribution
of this paper is to show how to speed up this verification.
 
Before entering that discussion, we briefly discuss the probability that
$\alpha$ be a normal element: if we write
$\alpha = a_0 + \cdots + a_{n-1} \xbar^{n-1}$, the determinant of
$\mat M_G(\alpha)$ is a (not identically zero) homogeneous polynomial of
degree $n$ in $(a_0,\dots,a_{n-1})$. If the $a_i$'s are chosen uniformly at
random in a finite set $X \subset \F$, the Lipton-DeMillo-Schwartz-Zippel
implies that the probability that $\alpha$ be normal is at least $1-n/|X|$.

If $G$ is cyclic, \cite{GatGie90} avoid computing a determinant by
computing the GCD of $S_\alpha := \sum_{i = 0}^{n-1} g_i(\alpha)x^i$
and $x^n-1$. In effect, this amounts to testing whether $S_\alpha$ is
invertible in the group ring $\K[G]$, which is isomorphic to
$\K[x]/\langle x^n-1\rangle$. This is a general fact: for any $G$,
$\mat M_G(\alpha)$ is the matrix of (left) multiplication by the orbit
sum
$$S_\alpha:= \sum_{g \in G} g(\alpha)g \in \K[G],$$ and $\alpha$ being
normal is equivalent to $S_\alpha$ being a unit in $\K[G]$. This point of
view may make it possible to avoid linear algebra of size $n$ over $\K$,
but writing $S_\alpha$ itself still involves $\Theta(n^2)$ elements in
$\F$. The following lemma is the main new ingredient in our algorithm: it
gives a randomized reduction to testing whether a suitable projection of
$S_\alpha$ in $\F[G]$ is a unit.
 

\begin{lemma}
  \label{Lem:Proj}
  For $\alpha \in \K$, $\mat M_G(\alpha)$ is invertible if and only
  if $$\ell(\mat M_G(\alpha)) := [\ell(g_ig_j(\alpha))]_{ij} \in M_n(\F)$$
  is invertible, for a generic $\F$-linear projection $\ell: \K \to \F$.
\end{lemma}
\begin{proof}
  $(\Rightarrow)$ For a fixed $\alpha\in\K$, any entry of
  $\mat M_G(\alpha)$ can be written as
  \begin{equation}\label{Eq:PrimElm}
    \sum_{k= 0}^{n-1} a_{ijk}\xbar^k,
  \end{equation}
  and for $\ell: \K \to \F$, the corresponding entry in $\ell(\mat
  M_G(\alpha))$ can be written $\sum_{k= 0}^{n-1} a_{ijk}\ell_k$, with
  $\ell_k = \ell(\xbar^k)$. Replacing these $\ell_k$'s by
  indeterminates $L_k$'s, the determinant becomes a polynomial in $P
  \in \F[L_1, \ldots, L_n].$ Viewing $P$ in $\K[L_1, \ldots, L_n]$, we
  have $ P(1, \xbar, \ldots, \xbar^{n-1})$ $= \det(\mat M_G(\alpha))$,
  which is non-zero by assumption. Hence, $P$ is not identically zero,
  and the conclusion follows.
  
  $(\Leftarrow)$ Assume $\mat M_G(\alpha)$ is not invertible. Following the
  proof of \cite[Lemma 4]{Jam18}, we first show that there exists a
  non-zero $\boldsymbol{u} \in \F^n$ in the kernel of $\mat M_G(\alpha)$.
  
  The elements of $G$ act on rows of $\mat M_G(\alpha)$ entrywise and the
  action permutes the rows the matrix. Assume
  $\varphi : G \to \mathfrak{S}_n$ is the group homomorphism such that
  $g(\mat M_i) = \mat M_{\varphi(g)(i)}$ for all $i$, where $\mat M_i$ is
  the $i$-th row of $\mat M_G(\alpha)$.
  
  Since $\mat M_G(\alpha)$ is singular, there exists a non-zero
  $\boldsymbol{v} \in \K^n$ such that $\mat M_G(\alpha)\boldsymbol{v} = 0$;
  we choose $\boldsymbol{v}$ having the minimum number of non-zero
  entries. Let $i \in \lbrace 1, \ldots , n \rbrace$ such that
  $v_i \neq 0$. Define $\boldsymbol{u} = 1/v_i\boldsymbol{v}$. Then,
  $\mat M_G(\alpha)\boldsymbol{u} = 0,$ which means
  $\mat M_j \boldsymbol{u} = 0 $ for $j \in \lbrace 1, \ldots, n
  \rbrace$. For $g \in G$, we have
  $g(\mat M_j \boldsymbol{u}) = \mat M_{\varphi(g)(j)} g(\boldsymbol{u})=
  0.$ Since this holds for any $j$, we conclude that
  $\mat M_G(\alpha)g(\boldsymbol{u})= 0$, hence
  $g(\boldsymbol{u})-\boldsymbol{u}$ is in the kernel of
  $\mat M_G(\alpha)$. On the other hand since the $i$-th entry of
  $\boldsymbol{u}$ is one, the $i$-th entry of
  $g(\boldsymbol{u}) -\boldsymbol{u}$ is zero. Thus the minimality
  assumption on $\textbf{v}$ shows that
  $g(\boldsymbol{u}) -\boldsymbol{u} = 0$, equivalently
  $g(\boldsymbol{u})=\boldsymbol{u}$, and hence $\boldsymbol{u} \in \F^n$.
  
  Now we show that $\ell(\mat M_G(\alpha))$ is not invertible for all
  choices of $\ell$. By Equation \eqref{Eq:PrimElm}, we can write
  $$\mat M_G(\alpha) = \sum_{j = 0}^{n-1} \mat M^{(j)} \xbar^j, \quad 
  \mat M^{(j)} \in M_{n}(\F) \text{~for all $j$}.$$ 
  Since $\boldsymbol{u}$ has entries in $\F$,
  $\mat M_G(\alpha) \boldsymbol{u} =0$ yields
  $\mat M^{(j)}\boldsymbol{u} = 0$ for
  $j \in \lbrace 1, \ldots , n \rbrace$. Hence,
$$\sum_{j = 0}^{n-1} \mat M^{(j)} \ell_j \boldsymbol{u} = 0$$ for any 
$\ell_j$'s in $\F$, and $\ell(\mat M_G(\alpha))$ is not invertible for any~$\ell$.
\end{proof} 
Our algorithm can be sketched as follows: choose random
$\alpha$ in $\K$ and $\ell: \K\to\F$, and let
\begin{equation}\label{def:s_alpha_ell}
s_{\alpha,\ell}:=\sum_{g \in G} \ell(g(\alpha))g \in \F[G].  
\end{equation}
 The
matrix $\ell(\mat M_G (\alpha))$ is the multiplication matrix by
$s_{\alpha,\ell}$ in $\F[G]$, so once $s_{\alpha,\ell}$ is known, we
are left with testing whether it is a unit in $\F[G]$.  In the next
two sections, we address the respective questions of computing
$s_{\alpha,\ell}$, and testing its invertibility in $\F[G]$.


\section{Computing projections of the orbit sum}
\label{sec:osum}

In this section we present algorithms to compute $s_{\alpha,\ell}$,
when $G$ is either abelian or metacyclic. We start by sketching our
ideas in simplest case, cyclic groups.  We will see that they follow
closely ideas used in \cite{KalSho98} over finite fields.

Suppose $G = \langle g \rangle$, so that given $\alpha$ in $\K$ and
$\ell: \K \to \F$, our goal is to compute
\begin{equation}
  \label{eq:cycproj}
  \ell(g^i(\alpha)), ~~\mbox{for}~ 0\leq i\leq n-1.
\end{equation}
\citeN{KalSho98} call this the \emph{automorphism projection problem} and
gave an algorithm to solve it in subquadratic time, when $g$ is the
$q$-power Frobenius $\mathbb{F}_{q^n} \to \mathbb{F}_{q^n}$.  The key idea in their
algorithm is to use the baby-steps/giant-steps technique: for a suitable
parameter $t$, the values in \eqref{eq:cycproj} can be rewritten as
\[
  (\ell \circ g^{tj})(g^i(\alpha)), ~~\mbox{for}~ 0 \leq j < m:=\lceil n/t
  \rceil ~\mbox{and}~ 0 \leq i <t.
\]
First, we compute all $G_i:=g^i(\alpha)$ for $0 \leq i <t$.  Then we compute
all $L_j:=\ell \circ g^{tj}$ for $0 \leq j <m$, where the $L_j$'s are
themselves linear mappings $\K \to \F$.  Finally, a matrix product yields
all values $L_j(G_i)$.

The original algorithm of \citeN{KalSho98} relies on the properties of the
Frobenius mapping to achieve subquadratic runtime. In our case, we cannot
apply these results directly; instead, we have to revisit the proofs
of~\citeN{KalSho98}, Lemmata 3, 4 and~8, now considering rectangular matrix
multiplication.  Our exponents involve the constant $\omega(4/3)$, for
which we have the upper bound $\omega(4/3) < 2.654$: this follows from the
upper bounds on $\omega(1.3)$ and $\omega(1.4)$ given by~\citeN{LeGall}, and
the fact that $k \mapsto \omega(k)$ is convex~\cite{LoRo83}. In particular,
$3/4 \cdot \omega(4/3) < 1.99$. Note also the inequality
$\omega(k) \ge 1+k$ for $k\ge 1$, since $\omega(k)$ describes products with
input and output size $O(n^{1+k})$.


\subsection{Multiple automorphism evaluation and applications}

The key to the algorithms below is the remark following
Assumption~\ref{assum}, which reduces automorphism evaluation to
modular composition of polynomials.  Over finite fields, this idea goes back
to~\citeN{GaSh92}, where it was credited to Kaltofen.

For instance, given $g \in G$ (by means of $\gamma:=g(\xbar)$), we can
deduce $g^2 \in G$ (again, by means of its image at $\xbar$) as
$\gamma(\gamma)$; this can be done with $\tilde{O}(n^{(\omega+1)/2})$
operations in $\F$ using Brent and Kung's modular composition
algorithm~\cite{BrKu78}. The algorithms below describe similar operations
along these lines, involving several simultaneous evaluations.

\begin{lemma}
  \label{lem:modcom}
  Given $\alpha_1,\dots,\alpha_s$ in $\K$ and $g$ in $G =
  \mathrm{Gal}(\K/\F)$, with $s = O(\sqrt{n})$, we can compute
  $g(\alpha_1),\dots,g(\alpha_s)$ with $\tilde
  O(n^{(3/4)\cdot\omega(4/3)})$ operations in $\F$.
\end{lemma}
\begin{proof}
(Compare \cite[Lemma~3]{KalSho98}) As noted above, for $i\le s$,
  $g(\alpha_i) = \alpha_i(\gamma)$, with $\gamma := g(\xbar) \in \K$.
  Let $t := \lceil n^{3/4} \rceil$, $m:=\lceil n/t\rceil$, and rewrite $\alpha_1 , \ldots , \alpha_s$ as 
$$\alpha_i = \sum_{0 \leq j < m} a_{i,j}\xbar^{tj},$$ where the
  $a_{i,j}$'s are polynomials of degree less than $t$. The next step
  is to compute $\gamma_i := \gamma^i$, for $i = 0 , \ldots , t$.
  There are $t$ products in $\K$ to perform, so this amounts to
  $\tilde{O}(n^{7/4})$ operations in $\F$.

  Having $\gamma_i$'s in hand, one can form the matrix
  $\boldsymbol{\Gamma} := \left[ \Gamma_0 ~ \cdots ~ \Gamma_{t-1}
    \right]^T$, where each column $\Gamma_i$ is the coefficient vector
  of $\gamma_i$ (with entries in $\F$); this matrix has $t \in
  O(n^{3/4})$ rows and $n$ columns. We also form
  $$\mat A := \left[{A}_{1,0} \cdots {A}_{1,m-1} \cdots
    {A}_{s,0} \cdots {A}_{s,m-1}\right]^T,$$ where
  ${A}_{i,j}$ is the coefficient vector of $a_{i,j}$. This matrix 
  has $s m \in O(n^{3/4})$ rows and $t \in O(n^{3/4})$ columns.

  Compute $\mat B:=\mat A\, \boldsymbol{\Gamma}$; as per our
  definition of exponents $\omega(\cdot )$, this can be done in
  $O(n^{(3/4)\cdot \omega(4/3)})$ operations in $\F$, and the rows of this matrix
  give all $a_{i,j}(\gamma)$.  The last step to get all
  $\alpha_i(\gamma)$ is to write them as $\alpha_i(\gamma) = \sum_{0
    \leq j < m} a_{i,j}(\gamma) \gamma_t^{j}.$ Using Horner's scheme,
  this takes $O(sm)$ operations in $\K$, which is $\tilde{O}(n^{7/4})$
  operations in $\F$. Since we pointed out that $\omega(3/4) \ge 7/4$,
  the leading exponent in all costs seen so far is
  $(3/4)\cdot\omega(4/3)$.
\end{proof}

\begin{lemma}\label{lem:selfcomp}
Given $\alpha$ in $\K$, $g_1, \ldots , g_{r}$ in $G =
\mathrm{Gal}(\K/\F)$ and positive integers $(s_1, \ldots s_r)$ such
that $\prod_{i = 1}^r s_i = O(\sqrt{n})$ and $r \in O(\log(n))$, all
  $$g_1^{i_1}\cdots g_r^{i_r}(\alpha) ,\quad \text{~for~} 0 \leq i_j
\leq s_j,\ 1 \leq j \leq r$$ can be computed in $\osumcosttilde$
operations in $\F$.
\end{lemma}
\begin{proof}
(Compare \cite[Lemma~4]{KalSho98}.) For a given  $m\in\{1,\dots,r\}$, suppose we have computed 
  $$G_{i_1,\dots,i_m}:=g_m^{i_m}\cdots g_1^{i_1}(\alpha)$$ for $0 \leq
  i_j \leq s_j$ if $1 \leq j < m$, and $0 \leq i_m < k_m,$ as well as
  the automorphism $\eta:={g_m}^{k_m}$ (by means of its value at $\xbar$, as per our convention).
  
 Then, we can obtain $G_{i_1,\dots,i_m}$ for $0 \leq i_j \leq s_j$ if $1
 \leq j < m$, and $0 \leq i_m < 2k_m$, by computing
 $\eta(G_{i_1,\dots,i_m})$, for all indices $i_1,\dots,i_m$ available
 to us, that is, $0 \leq i_j \leq s_j$ if $1 \leq j < m$, and $0 \leq
 i_m < k_m$. This can be carried out using $\osumcosttilde$ operations
 in $\F$ by applying Lemma \ref{lem:modcom}. Prior to entering the
 next iteration, we also compute $\eta^2$ by means of one modular
 composition, whose cost is negligible. 

 Using the above doubling method for $g_m$, we have to do $O(\log
 s_m)$ steps, for a total cost of $\osumcosttilde$ operations in $\F$.  We
 repeat this procedure for $m=1,\dots,r$; since $r$ is in $O(\log(n))$,
 the cost remains $\osumcosttilde$.
\end{proof}

We now present dual versions of the previous two lemmas (our
reference~\cite{KalSho98} also had such a discussion). Seen as an
$\F$-linear map, the operator $g:\alpha \mapsto g(\alpha)$ admits a
transpose, which maps an $\F$-linear form $\ell:\K\to\F$ to the
$\F$-linear form $\ell \circ g: \alpha \mapsto \ell(g(\alpha))$.  The
{\em transposition principle}~\cite{KaKiBs88,CaKaYa89} implies that if
a linear map $\F^N \to \F^M$ can be computed in time $T$, its
transpose can be computed in time $T+O(N+M)$. In particular, given $s$
linear forms $\ell_1,\dots,\ell_s$ and $g$ in $G$, transposing
Lemma~\ref{lem:modcom} shows that we can compute $\ell_1 \circ
g,\dots,\ell_s \circ g$ in time $\osumcosttilde$. The following lemma
sketches the construction.

\begin{lemma}
  \label{lem:modcomT}
  Given $\F$-linear forms $\ell_1,\dots,\ell_s:\K\to \F$ and $g$ in $G =
  \mathrm{Gal}(\K/\F)$, with $s = O(\sqrt{n})$, we can compute
  $\ell_1\circ g,\dots,\ell_s \circ g$ in time $\tilde
  O(n^{{3}/{4}\omega({4}/{3})})$.
\end{lemma}
\begin{proof}
  Given $\ell_i$ by its values on the power basis $1,\xbar,\dots,\xbar^{n-1}$, $\ell_i \circ g$ is represented by its values at
  $1,\gamma,\dots,\gamma^{n-1}$, with $\gamma := g(\xbar)$. 

  Let $t,m$ and $\gamma_0,\dots,\gamma_t$ be as in the proof of
  Lemma~\ref{lem:modcom}. Next, compute the ``giant steps''
  $\gamma_t^j = \gamma^{tj}$, $j=0,\dots,m-1$ and for $i=1,\dots,s$
  and $j=0,\dots,m-1$, deduce the linear forms $L_{i,j}$ defined by
  $L_{i,j}(\alpha) := \ell_i(\gamma^{tj}\alpha)$ for all $\alpha$ in
  $\K$. Each of them can be obtained by a {\em transposed
    multiplication} in time $\tilde{O}(n)$~\cite[Section~4.1]{Shoup},
  so that the total cost thus far is $\tilde{O}(n^{7/4})$.

  Finally, multiply the $(sm \times n)$ matrix with entries the
  coefficients of all $L_{i,j}$ (as rows) by the $(n \times t)$ matrix with
  entries the coefficients of $\gamma_0,\dots,\gamma_{t-1}$ (as columns) to
  obtain all values $\ell_i(\gamma^j)$, for $i=1,\dots,s$ an
  $j=0,\dots,n-1$.  This can be accomplished with
  $O(n^{(3/4)\cdot\omega(4/3)})$ operations in~$\F$.
\end{proof}

From this, we deduce the transposed version of Lemma~\ref{lem:selfcomp},
whose proof follows the same pattern.

\begin{lemma}
  \label{lem:transmodcomp}
  Given $\ell:\K\to F$, $g_1, \ldots , g_{r}$ in $G = \mathrm{Gal}(\K/\F)$
  and positive integers $(s_1, \ldots s_r)$ such that
  $\prod_{i = 1}^r s_i = O(\sqrt{n})$ and $r \in O(\log(n))$, all linear
  maps
  \[
    \ell \circ g_1^{i_1}\cdots g_r^{i_r} ,\quad \text{~for~} 0 \leq i_j
    \leq s_j,\ 1 \leq j \leq r,
  \]
  can be computed in $\osumcosttilde$ operations in $\F$.
\end{lemma} 
\begin{proof}
(Compare \cite[Lemma~8]{KalSho98}.)
  We proceed as in Lemma~\ref{lem:selfcomp}. For $m=1,\dots,r$, assume
  we know $L_{i_1,\dots,i_m}:=\ell \circ (g_1^{i_1}\cdots g_m^{i_m}),$
  for $0 \leq i_j \leq s_j$ if $1 \leq j < m$, and $0 \leq i_m < k_m.$
  Using the previous lemma, we compute all $L_{i_1,\dots,i_m} \circ
  {g_m}^{k_m},$ which gives us $L_{i_1,\dots,i_m}$ for indices $0 \le
  i_m < 2k_m$. The cost analysis is as in Lemma~\ref{lem:selfcomp}.
\end{proof}



\subsection{Abelian Groups}
\label{ssec:proj_abelian}

The first main result in this section is the following proposition.
Assume $G$ is an abelian group presented as 
\[
  \langle g_1, \ldots , g_r: g_{1}^{e_1} = \cdots = g_{r}^{e_r} = 1
  \rangle,
\]
where $ e_i \in \mathbb{N}$ is the order of $g_i$ and $n = e_1 \cdots
e_r$.  Without loss of generality, we assume $e_i \ge 2$ for all $i$,
so that $r$ is in $O(\log n)$. Elements of $\F[G]$ are written as
polynomials $\sum_{i_1,\dots,i_r} c_{i_1,\dots,i_r}
{g_1}^{e_1} \cdots {g_r}^{e_r}$, with $0\le i_j < e_j$ for all~$j$.

\begin{proposition}\label{prop:abelian}
  Suppose that $G$ is abelian, with notation as above. For $\alpha$ in
  $\K$ and $\ell:\K\to\F$, $s_{\alpha,\ell} \in \F[G]$, as defined
  in~\eqref{def:s_alpha_ell}, is computable using $\osumcosttilde$
  operations in $\F$.
\end{proposition}
\begin{proof}
Our goal is to compute
\begin{equation}\label{eq:abelian}
  \ell (g_1^{i_1},  \ldots, g_r^{i_r}(\alpha)), \, 1 \leq j \leq r, 0 \leq i_j \leq e_j,
\end{equation}
where $\ell$ is an $\F$-linear projection $\K\to \F$.  For $ 1\leq i
\leq r$, define $s_i:=\lceil\sqrt{e_i} \rceil$. As we sketched in the
cyclic case, the elements in \eqref{eq:abelian} can be expressed as
$L_{j_1,\dots, j_r} (G_{i_1,\dots,i_r})$, 
for $1\leq m \leq r, 0\leq i_m < s_m, 0 \leq j_m < s_m$.
Here, $L_{j_1,\dots,j_r} :=\ell \circ (g_1^{s_1j_1} \cdots
g_r^{s_rj_s})$ are linear projections presented as row vectors and
$G_{i_1,\dots,i_r}:=g_1^{i_1} \cdots g_r^{i_r}(\alpha)$ are field
elements presented as column vectors. Then, all elements in
\eqref{eq:abelian} can be computed with the following steps, the sum of whose 
costs proves the proposition.

\smallskip\noindent \textbf{Step 1.} Apply Lemma \ref{lem:selfcomp} to get 
$$G_{i_1,\dots,i_r}=g_1^{i_1} \cdots g_r^{i_r}(\alpha), \,\, 1\leq m \leq r, 0\leq i_m < s_m,$$
with cost $\osumcosttilde$.

\smallskip\noindent\textbf{Step 2.} Compute all $g_i^{s_i}$, $i=1,\dots,r$;
this involves $O(\log(n))$ modular compositions, so the cost is negligible
compared to that of Step~1.

\smallskip\noindent\textbf{Step 3.} Use Lemma \ref{lem:transmodcomp} to compute 
$$L_{j_1,\dots,j_r} = \ell \circ (g_1^{s_1j_1} \cdots
g_r^{s_rj_s}), \,\, 1\leq m \leq r, 0 \leq j_m < s_m,$$
with cost $\osumcosttilde$

\smallskip\noindent\textbf{Step 4.} Multiply the matrix with rows the
coefficients of all $L_{j_1,\dots,j_r}$ by the matrix with columns the
coefficients of all $G_{i_1,\dots,i_r}$; this yields all required
values. We compute this product in $O(n^{(1/2)\cdot\omega(2)})$
operations in $\F$, which is in $\osumcost$.
\end{proof}


\subsection{Metacyclic Groups}

A group $G$ is metacyclic if it has a normal cyclic subgroup $H$ such that
$G/H$ is cyclic; for instance, any group with a squarefree order is
metacyclic. See \cite[p.~88]{Johnson} or \cite[p.~334]{Curtis} for more
background. A metacyclic group can always be presented~as
\begin{equation}
  \label{eq:metacyclic}
  \langle \sigma,\tau: \sigma^m = 1,  \tau^s = \sigma^t, \tau^{-1}\sigma \tau = \sigma^r \rangle,
\end{equation}
for some integers $m,t,r,s$, with $r,t \leq m$ and
$r^s = 1 \bmod t, rt = t \bmod m$. For example, the dihedral group
$$D_{2m} = \langle \sigma,\tau: \sigma^m =1, \tau^2 = 1, \tau^{-1}
\sigma \tau = \sigma^{m-1} \rangle, $$ is metacyclic, with
$s=2$. Generalized quaternion groups, which can be presented as
$$Q_m = \langle \sigma,\tau: \sigma^{2m} =1, \tau^2 = \sigma^m,
\tau^{-1} \sigma \tau = \sigma^{2m-1} \rangle,$$ are metacyclic, with
$s=2$ as well.

Using the notation of~\eqref{eq:metacyclic}, $n=|G|$ is equal to $ms$, and
all elements in a metacyclic group can be presented uniquely as either
\begin{equation}\label{pres1}
\{\sigma^i \tau^j,\,\,\, 0\leq i \leq m-1,\ 0\leq j \leq s-1\}  
\end{equation}
or
\begin{equation}\label{pres2}
\{ \tau^j\sigma^i,\,\,\, 0\leq i \leq m-1,\ 0\leq j \leq s-1\}.
\end{equation}
Accordingly, elements in the group algebra $\F[G]$ can be written as 
either 
$$\sum_{\substack{i <m\\ j< s}} c_{i,j} \sigma^i \tau^j \quad\text{or}\quad
\sum_{\substack{i <m\\ j< s}} c'_{i,j} \tau^j \sigma^i.$$
Conversion between the two representations involves no operation in $\F$,
using the commutation relation $\sigma^k \tau^c = \tau^c \sigma^{kr^c}$
for $k,c \ge 0$.

\begin{proposition}
  Suppose that $G$ is metacyclic, with notation as above. For $\alpha$
  in $\K$ and $\ell:\K\to\F$, $s_{\alpha,\ell} \in \F[G]$ is
  computable using $\osumcosttilde$ operations in $\F$.
\end{proposition}
\begin{proof}
  Suppose first that $s \le m$; then, we use the
  presentation~\eqref{pres1} of the elements of $G$. Take $\alpha$ in
  $\K$ and $\ell:\K \to \F$; the goal is to compute
  $\ell(\sigma^i\tau^j (\alpha))$, for all $0\leq i < m$ and $0 \leq j
  <s.$ This is accomplished with the following steps.

\smallskip\noindent\textbf{Step 1.} Apply Lemma \ref{lem:selfcomp} to compute
$$G_{i,j} := \sigma^i\tau^j(\alpha),\ 0\leq i < \lceil \sqrt{m/s}
\rceil,\ 0 \leq j < s.$$ Note that
$\lceil \sqrt{m/s} \rceil s \leq \lceil \sqrt{sm} \rceil \in O(\sqrt n)$,
so we are under the assumptions of the lemma. This takes $\osumcosttilde$
operations in~$\F$.

\smallskip\noindent\textbf{Step 2.} Compute
$\sigma^{\lceil \sqrt{m/s} \rceil}$, in $O(\log(n))$ modular compositions
in degree $n$. The cost is no more than that of Step 1.

\smallskip\noindent\textbf{Step 3.} Compute
\[
  L_k := \ell \circ \sigma^{k\lceil \sqrt{m/s} \rceil}, \,\, 0\leq k <
  \lceil \sqrt{sm}\rceil,
\]
using Lemma \ref{lem:transmodcomp}.  This takes $\osumcosttilde$ operations
in~$\F$.

\smallskip\noindent\textbf{Step 4.} At this point, we compute all
$$ L_k(s_{i,j}) = \ell(\sigma^{k\lceil \sqrt{m/s} \rceil + i}\tau
^j(\alpha)),$$ for $0\le k < \lceil \sqrt{sm}\rceil$,
$0\le i< \lceil \sqrt{m/s}\rceil$ and $0 \leq j < s;$ these are precisely
the values we needed.

This can be carried out by multiplying the matrix with rows the
coefficients of all $L_k$ by the matrix with columns the coefficients of
all $G_{i,j}$; this yields all required values, as pointed out above. There
are $O(\sqrt{sm})=O(\sqrt{n})$ linear forms $L_k$'s, and $O(\sqrt{n})$
field elements $G_{i,j}$'s, so we can compute this product in
$O(n^{(1/2)\cdot\omega(2)})$ operations in $\F$, which is $\osumcost$.

This concludes the proof in the case $s \le m$. When $m \le s$, use the
presentation~\eqref{pres2} of the elements of $G$ and proceed as above.
\end{proof}


 


\section{Testing Invertibility in the Group Algebra}
\label{sec:invertibility}

In this section we consider the problem of invertibility testing in
$\F[G]$, specifically for abelian and metacyclic groups $G$: given an
element $\beta$ in $\F[G]$, for a field $\F$ and a group $G$, determine
whether $\beta$ is a unit in $\F[G]$.  As well as being necessary in our
application to normal bases, we believe these problems are of independent
interest.

Since we are in characteristic zero, Wedderburn's theorem implies the
existence of an $\F$-algebra isomorphism (which we will refer to as a
Fourier Transform)
\[
  \F[G] \to M_{d_1}(D_1) \times \dots \times M_{d_r}(D_r),
\]
where all $D_i$'s are division algebras over $\F$. If we were working over
$\F=\C$, all $D_i$'s would simply be $\C$ itself.  A natural solution
to test the invertibility of $\beta \in \F[G]$ would then be to compute its
Fourier transform and test whether all its components
$\beta_1 \in M_{d_1}(\C),\dots,\beta_r \in M_{d_r}(\C)$ are
invertible. This boils down to linear algebra over $\C$, and takes
$O(d_1^\omega + \cdots + d_r^\omega)$ operations.  Since
$d_1^2 + \cdots + d_r^2 = |G|$, this is $O(|G|^{\omega/2})$ operations in
$\C$.

However, we do not wish to make such a strong assumption as $\F=\C$. Since
we measure the cost of our algorithms in $\F$-operations, the direct
approach that embeds $\F[G]$ into $\C[G]$ does not make it possible to
obtain a subquadratic cost in general. If, for instance, $\F=\Q$ and $G$ is
cyclic of order $n=2^k$, computing the Fourier Transform of $\beta$
requires we work in a degree $n/2$ extension of $\Q$, implying a quadratic
runtime.

In this section, we give algorithms for the problem of invertibility
testing for the families of group seen so far, abelian and metacyclic. For
the former, we actually prove a stronger result: starting from a suitable
presentation of $G$, we give a softly linear-time algorithm to find an
isomorphic image of $\beta \in \F[G]$ in a product of $\F$-algebras of the
form $\F[z]/\langle P_i(z)\rangle$, for certain polynomials $P_i \in \F[z]$
(recovering $\beta$ from its image is softly-linear time as well). Not only
does this allow us to test whether $\beta$ is invertible, this would also
make it possible to find its inverse in $\F[G]$ (or to compute products in
$\F[G]$) in softly-linear time.  We are not aware of previous results of
this kind. For metacyclic groups, we describe an injective $\F$-algebra
homomorphism from $\F[G]$ to a matrix algebras over a cyclotomic ring. The
codomain is in general of dimension higher than $|G|$, so the algorithm we
deduce from this is not linear-time.


\subsection{Abelian groups}

Because an abelian group is a product of cyclic groups, the group
algebra $\F[G]$ of such a group is the tensor product of cyclic
algebras. Using this property, given an element $\beta$ in $\F[G]$,
our goal in this section is to determine whether $\beta$ is a unit.

The previous property implies that $\F[G]$ admits a description of the
form $\F[x_1,\dots,x_t]/\langle x_1^{n_1}-1,\dots,x_t^{n_t}-1\rangle$,
for some integers $n_1,\dots,n_t$. The complexity of arithmetic
operations in an $\F$-algebra such as $\A:=\F[x_1,\dots,x_t]/\langle
P_1(x_1),\dots,P_t(x_t)\rangle$ is difficult to pin down precisely. For
general $P_i$'s, the cost of multiplication in $\A$ is known to be
$O(\dim(\A)^{1+\varepsilon})$, for any $\varepsilon >
0$~\cite[Theorem~2]{LiMoSc09}. From this it may be possible to deduce
similar upper bounds on the complexity of invertibility tests,
following~\cite{DaMMMScXi06}, but this seems non-trivial.

Instead, we give an algorithm with softly linear runtime, that uses
the factorization properties of cyclotomic polynomials and Chinese
remaindering techniques to transform our problem into that of
invertibility testing in algebras of the form $\F[z]/\langle P_i(z)
\rangle$, for various polynomials $P_i$. The reference~\cite{Pol94}
also discusses the factors of algebras such as
$\F[x_1,\dots,x_t]/\langle x_1^{n_1}-1,\dots,x_t^{n_t}-1\rangle$, but
the resulting algorithms are different (and the cost of the
\citeauthor{Pol94}'s \citeyear{Pol94} algorithm is only known to be
polynomial in $|G|$).

\smallskip

\noindent{\bf Tensor product of two cyclotomic rings: coprime orders.}
The following proposition will be the key to foregoing multivariate
polynomials, and replacing them by univariate ones.  Let $m,m'$ be two
coprime integers and define
$$\mathbbm{h}:=\F[x,x']/\langle \Phi_{m}(x), \Phi_{m'}(x')\rangle,$$
where for $i \ge 0$, $\Phi_i$ is the cyclotomic polynomial of order
$i$. In what follows, $\varphi$ is Euler's totient function, so that
$\varphi(i) = \deg(\Phi_i)$ for all~$i$.
\begin{lemma}
  There exists an $\F$-algebra isomorphism $\gamma: \mathbbm{h} \to
  \F[z]/\langle\Phi_{mm'}(z)\rangle$ given by $xx' \mapsto z$.  Given
  $\Phi_m$ and $\Phi_{m'}$, $\Phi_{mm'}$ can be computed in time
  $\tilde{O}(\varphi(mm'))$; given these polynomials, one can
  apply $\gamma$ and its inverse to any input using
  $\tilde{O}(\varphi(mm'))$ operations in~$\F$.
\end{lemma}
\begin{proof}
  Without loss of generality, we prove the first claim over $\Q$; the
  result over $\F$ follows by scalar extension. In the field \sloppy
  $\Q[x,x']/\langle \Phi_{m}(x), \Phi_{m'}(x')\rangle$, $xx'$ is
  cancelled by $\Phi_{mm'}$. Since this polynomial is irreducible, it
  is the minimal polynomial of $xx'$, which is thus a primitive
  element for $\Q[x,x']/\langle \Phi_{m}(x),
  \Phi_{m'}(x')\rangle$. This proves the first claim.

  For the second claim, we first determine the images of $x$ and $x'$
  by $\gamma$. Start from a B\'ezout relation $am+ a'm'=1$, for some
  $a,a'$ in $\Z$.  Since $x^m = {x'}^{m'}=1$ in $\mathbbm{h}$, we
  deduce that $\gamma(x)=z^{u}$ and $\gamma(x') = z^{v}$, with $u:=am
  \bmod mm'$ and $v:=a'm' \bmod mm'$. To compute $\gamma(P)$, for some
  $P$ in $\mathbbm{h}$, we first compute $P(z^u, z^v)$, keeping all
  exponents reduced modulo $mm'$. This requires no arithmetic
  operations and results in a polynomial $\bar P$ of degree less than
  $mm'$, which we eventually reduce modulo $\Phi_{mm'}$ (the latter is
  obtained by the composed product algorithm of~\cite{BoFlSaSc06} in
  quasi-linear time).  By~\cite[Theorem~8.8.7]{BacSha96}, we have the
  bound $s \in O(\varphi(s) \log(\log(s)))$, so that $s$ is in
  $\tilde{O}(\varphi(s))$. Thus, we can reduce $\bar P$ modulo
  $\Phi_{mm'}$ in $\tilde{O}(\varphi(mm'))$ operations, establishing
  the cost bound for $\gamma$.

  Conversely, given $Q$ in $\F[z]/\langle\Phi_{mm'}(z)\rangle$, we obtain
  its preimage by replacing powers of $z$ by powers of $xx'$, reducing all
  exponents in $x$ modulo $m$, and all exponents in $x'$ modulo $m'$.  We
  then reduce the result modulo both $\Phi_m(x)$ and $\Phi_{m'}(x')$.  By
  the same argument as above, the cost is softly linear in $\varphi(mm')$.
\end{proof}

\noindent{\bf Extension to several cyclotomic rings.}  The natural
generalization of the algorithm above starts with pairwise distinct
primes $\boldsymbol{p}=(p_1,\dots,p_t)$, non-negative exponent
$\boldsymbol{c}=(c_1,\dots,c_t)$ and variables
$\boldsymbol{x}=(x_1,\dots,x_t)$ over $\F$. Now, we define
$$\H:=\F[x_1,\dots,x_t]/\langle
\Phi_{{p_1}^{c_1}}(x_1),\dots,\Phi_{{p_t}^{c_t}}(x_t)\rangle;$$ when
needed, we will write $\H$ as
$\H_{\boldsymbol{p},\boldsymbol{c},\boldsymbol{x}}$. Finally, we let
$\mu:={p_1}^{c_1}\cdots {p_t}^{c_t}$; then, the dimension $\dim(\H)$ is
$\varphi(\mu)$.

\begin{lemma}\label{lemma:distinctP}
 There exist an $\F$-algebra isomorphism $\Gamma: \H \to
 \F[z]/\langle\Phi_{\mu}(z)\rangle$ given by $x_1 \cdots x_t \mapsto
 z$.  One can apply $\Gamma$ and its inverse to any input using
 $\tilde{O}(\dim(\H))$ operations in $\F$.
\end{lemma}
\begin{proof}
  We proceed iteratively. First, note that the cyclotomic polynomials
  $\Phi_{{p_i}^{c_i}}$ can all be computed in time $O(\varphi(\mu))$. 
  The isomorphism
  $\gamma: \F[x_1,x_2]/\langle \Phi_{{p_1}^{c_1}}(x_1),
  \Phi_{{p_2}^{c_2}}(x_2)\rangle \to \F[z]/\langle
  \Phi_{{p_1}^{c_1}{p_2}^{c_2}}(z)\rangle$
given in the previous paragraph extends coordinate-wise to an
  isomorphism
  $$\Gamma_1: \H \to \F[z,x_3,\dots,x_t]/\langle
  \Phi_{{p_1}^{c_1}{p_2}^{c_2}}(z),\Phi_{{p_3}^{c_3}}(x_3),\dots,\Phi_{{p_t}^{c_t}}(x_t)\rangle.$$
  By the previous lemma, $\Gamma_1$ and its inverse can be applied to
  any input in time $\tilde{O}(\varphi(\mu))$. Iterate this process
  another $t-2$ times, to obtain $\Gamma$ as a product
  $\Gamma_{t-1} \circ \cdots \circ \Gamma_1$. Since $t$ is logarithmic 
  in $\varphi(\mu)$, the proof is complete.
\end{proof}

\noindent{\bf Tensor product of two prime-power cyclotomic rings, same $p$.}~In the following two paragraphs, we discuss the opposite situation as
above: we now work with cyclotomic polynomials of prime power
orders for a common prime $p$. As above, we start with two such polynomials.

Let thus $p$ be a prime. The key to the following algorithms is the
lemma below.  Let $c,c'$ be positive integers, with $c \ge
c'$, and let $x,y$ be indeterminates over $\F$. Define
\begin{align}
\mathbbm{a}&:=\F[x]/\Phi_{p^c}(x),  \\
\mathbbm{b}&:=\F[x,y]/\langle \Phi_{p^c}(x), \Phi_{p^{c'}}(y)\rangle = \mathbbm{a}[y]/\Phi_{p^{c'}}(y).
\end{align}
Note that $\mathbbm{a}$ and $\mathbbm{b}$ have respective dimensions
$\varphi(p^c)$ and $\varphi(p^c) \varphi(p^{c'})$.
\begin{lemma}
  There is an $\F$-algebra isomorphism $\theta: \mathbbm{b} \to
  \mathbbm{a}^{\varphi(p^{c'})}$ such that one can apply $\theta$ or
  its inverse to any inputs using $\tilde{O}(\dim(\mathbbm{b}))$ operations in $\F$.
\end{lemma}
\begin{proof}
  Let $\xbar$ be the residue class of
  $x$ in $\A$. Then, in $\mathbbm{a}[y]$, $\Phi_{p^{c'}}(y)$ factors as
  $$\Phi_{p^{c'}}(y) =\prod_{\substack{1 \le i\le p^{c'}-1\\ \gcd(i,p)
      =1}} (y-\rho_i),$$ with $\rho_i:={\xbar}^{i p^{c-c'}}$ for all
  $i$.  Even though $\mathbbm{a}$ may not be a field, the Chinese
  Remainder theorem implies that $\mathbbm{b}$ is isomorphic to
  $\mathbbm{a}^{\varphi(p^{c'})}$; the isomorphism is given by
  $$\begin{array}{cccc}
    \theta: & \mathbbm{b} & \to & \mathbbm{a} \times \cdots \times \mathbbm{a}, \\
    & P & \mapsto& (P(\xbar,\rho_1),\dots,P(\xbar,\rho_{\varphi(p^{c'})}).
  \end{array}$$
  In terms of complexity, arithmetic operations $(+,-,\times)$ in
  $\mathbbm{a}$ can all be done in $\tilde{O}(\varphi(p^c))$ operations
  in $\F$. Starting from $\rho_1 \in \mathbbm{a}$, all other roots
  $\rho_i$ can then be computed in $O(\varphi(p^{c'}))$ operations in
  $\mathbbm{a}$, that is, $\tilde{O}(\dim(\mathbbm{b}))$
  operations in $\F$. 
  
Applying $\theta$ and its inverse is done by means of fast evaluation
and interpolation~\cite[Chapter~10]{vzGathen13} in $\tilde{O}(\varphi(p^{c'}))$
operations in $\mathbbm{a}$, that is, $\tilde{O}(\deg(\mathbbm{b}))$ operations in $\F$
(the algorithms do not require that $\mathbbm{a}$ be a field).
\end{proof}

\smallskip\noindent{\bf Extension to several cyclotomic rings.}
Let $p$ be as before, and consider now non-negative integers
$\boldsymbol{c}=(c_1,\dots,c_t)$ and variables $\boldsymbol{x}=(x_1,\dots,x_t)$. We
define the $\F$-algebra
$$\A:=\F[x_1,\dots,x_t]/\langle \Phi_{p^{c_1}}(x_1), \dots,
\Phi_{p^{c_t}}(x_t)\rangle,$$ which we will sometimes write
$\A_{p,\boldsymbol{c},\boldsymbol{x}}$ to make the dependency on $p$
and the $c_i$'s clear. Up to reordering the $c_i$'s, we can assume
that $c_1 \ge c_i$ holds for all $i$, and define as before
$\mathbbm{a}:=\F[x_1]/\Phi_{p^{c_1}}(x_1)$.

\begin{lemma}\label{lemma:A}
  There exists an $\F$-algebra isomorphism $\Theta: \A \to
  \mathbbm{a}^{\dim(\A)/\dim(\mathbbm{a})}$. This isomorphism and its
  inverse can be applied to any inputs using $\tilde{O}(\dim(\A))$
  operations in $\F$.
\end{lemma}
\begin{proof}
Without loss of generality, we can assume that all $c_i$'s are non-zero
(since for $c_i=0$, $\Phi_{p^{c_i}}(x_i)=x_i-1$,
so $\F[x_i]/\langle \Phi_{p^{c_i}}(x_i) \rangle = \F$).
We proceed iteratively. First, rewrite $\A$ as
$$\A=\mathbbm{a}[x_2,x_3,\dots,x_t]/\langle \Phi_{p^{c_2}}(x_2), \Phi_{p^{c_3}}(x_3), \dots,
\Phi_{{p_t}^{c_t}}(x_t)\rangle.$$ 
The isomorphism 
$\theta: \mathbbm{a}[x_2]/\Phi_{p^{c_2}}(x_2) \to \mathbbm{a}^{\varphi(p^{c_2})}$
introduced in the previous paragraph extends coordinate-wise
to an isomorphism 
$$\Theta_1: \A \to (\mathbbm{a}[x_3,\dots,x_t]/\langle
\Phi_{p^{c_3}}(x_3), \dots,
\Phi_{p^{c_t}}(x_t)\rangle)^{\varphi(p^{c_2})};$$ $\Theta_1$ and its
inverse can be evaluated in quasi-linear time $\tilde{O}(\dim(\A))$.
We now work in all copies of $\mathbbm{a}[x_3,\dots,x_t]/\langle
\Phi_{p^{c_3}}(x_3), \dots, \Phi_{p^{c_t}}(x_t)\rangle$ independently,
and apply the procedure above to each of them. Altogether we have
$t-1$ such steps to perform, giving us an isomorphism
$$\Theta = \Theta_{t-1} \circ \cdots \circ \Theta_1:
\A \to
\mathbbm{a}^{\varphi(p^{c_2}) \cdots \varphi(p^{c_t})}.$$
The exponent can be rewritten as $ \dim(\A)/\dim(\mathbbm{a})$, as claimed.
In terms of complexity, all $\Theta_i$'s and their inverses can be computed
in quasi-linear time $\tilde{O}(\dim(\A))$, and we do $t-1$ of them,
where $t$ is $O(\log(\dim(\A)))$. 
\end{proof}

\noindent{\bf Decomposing certain $p$-group algebras.}  The prime $p$
and indeterminates $\boldsymbol{x}=(x_1,\dots,x_t)$ are as before; we now consider
positive integers $\boldsymbol{b}=(b_1,\dots,b_t)$, and the $\F$-algebra
\[
\begin{array}{ccl}
\B&:=&\F[x_1,\dots,x_t]/\langle x_1^{p^{b_1}}-1,\dots,x_t^{p^{b_t}}-1\rangle\\$$
&=& \F[x_1]/\langle x_1^{p^{b_1}}-1 \rangle \otimes \cdots \otimes \F[x_t]/\langle x_t^{p^{b_t}}-1 \rangle.
\end{array}
\]
If needed, we will write $\B_{p,\boldsymbol{b},\boldsymbol{x}}$ to make the dependency
on $p$ and the $b_i$'s clear. This is the $\F$-group algebra
of $\Z/p^{b_1}\Z \times \cdots \times \Z/p^{b_t}\Z$.

\begin{lemma}\label{lemma:alg}
  There exists a positive integer $N$, non-negative integers
  $\boldsymbol{c}=(c_1,\dots,c_N)$ and  an
  $\F$-algebra isomorphism 
  $$\Lambda: \B \to \D= \F[z]/\langle \Phi_{p^{c_1}}(z) \rangle \times \cdots \times \F[z]/\langle \Phi_{p^{c_N}}(z)\rangle.$$
  One can apply the isomorphism and its inverse to any 
  input using $\tilde{O}(\dim(\B))$ operations in $\F$.
\end{lemma}
\begin{proof}
For $i \le t$, we have the factorization
$$x_i^{p^{b_i}}-1 = \Phi_1(x_i) \Phi_p(x_i) \Phi_{p^2}(x_i) \cdots
\Phi_{p^{b_i}}(x_i);$$ note that $\Phi_1(x_i)=x_i-1$.  The factors may
not be irreducible, but they are pairwise coprime, so that we have a
Chinese Remainder isomorphism
\[
  \lambda_i: \F[x_i]/\langle x_i^{p^{b_i}}-1 \rangle \to \F[x_i]/\langle \Phi_1(x_i)\rangle
  \times \cdots \times  \F[x_i]/\langle \Phi_{p^{b_i}}(x_i)\rangle.
\]
Together with its inverse, this can be computed  
in $\tilde{O}(p^{b_i})$ operations in $\F$~\cite[Chapter~10]{vzGathen13}. By distributivity of the tensor
product over direct products, 
this gives an $\F$-algebra isomorphism
$$\lambda: \B \to \prod_{c_1=0}^{b_1} \cdots \prod_{c_t=0}^{b_t} \A_{p,\boldsymbol{c},\boldsymbol{x}},$$
with $\boldsymbol{c}=(c_1,\dots,c_t)$. Together with its inverse, 
$\lambda$ can be computed in $\tilde{O}(\dim(\B))$ operations in $\F$.
Composing with the result in Lemma~\ref{lemma:A}, this gives
us an isomorphism
$$\Lambda: \B \to \D:=\prod_{c_1=0}^{b_1} \cdots \prod_{c_t=0}^{b_t}
\mathbbm{a}_{\boldsymbol{c}}^{D_{\boldsymbol{c}}},$$ where
$\mathbbm{a}_{\boldsymbol{c}} = \F[z]/\langle \Phi_{p^c}(z)\rangle$,
with $c =\max(c_1,\dots,c_t)$ and $D_{\boldsymbol{c}} =
\dim(\A_{t,\boldsymbol{c},\boldsymbol{x}})/\dim(\mathbbm{a}_{\boldsymbol{c}})$. As
before, $\Lambda$ and its inverse can be computed in quasi-linear time
$\tilde{O}(\dim(\B))$.
\end{proof}
As for $\B$, we will write $\D_{p,\boldsymbol{b},\boldsymbol{x}}$ if needed; it is
well-defined, up to the order of the factors.

\smallskip

\noindent{\bf Main result.} Let $G$ be an abelian group.  We can
write the elementary divisor decomposition of $G$ as $G = G_1 \times
\cdots \times G_s$, where each $G_i$ is of prime power order
$p_i^{a_i}$, for pairwise distinct primes $p_1,\dots,p_s$, so that
$|G| = p_1^{a_1} \cdots p_s^{a_s}$. Each $G_i$ can itself be written
as a product of cyclic groups, $G_i = G_{i,1} \times \cdots \times
G_{i,t_i}$, where the factor $G_{i,j}$ is cyclic of order
${p_i}^{b_{i,j}}$, with $b_{i,1} \le \cdots \le b_{i,t_i}$; this is
the invariant factor decomposition of $G_i$, with $b_{i,1} + \cdots +
b_{i,t_i} = a_i$.

We henceforth assume that generators
$\gamma_{1,1},\dots,\gamma_{s,t_s}$ of respectively
$G_{1,1},\dots,G_{s,t_s}$ are known, and that elements of $\F[G]$ are
given on the power basis in $\gamma_{1,1},\dots,\gamma_{s,t_s}$. Were
this not the case, given arbitrary generators $g_1,\dots,g_r$ of $G$, with
orders $e_1,\dots,e_r$, a brute-force solution would factor each $e_i$
(factoring $e_i$ takes $o(e_i)$ bit operations on a standard RAM), so
as to write $\langle g_i \rangle$ as a product of cyclic groups of
prime power orders, from which the required decomposition follows.

\begin{proposition}
  Given $\beta \in \F[G]$, written on the power basis
  $\gamma_{1,1},\dots,\gamma_{s,t_s}$, one can test if $\beta$ is a
  unit in $\F[G]$ using $\tilde{O}(|G|)$ operations in $\F$.
\end{proposition}
The proof occupies the rest of this paragraph.
From the factorization $G = G_1 \times \cdots \times G_s$, we deduce
that the group algebra $\F[G]$ is the tensor product $\F[G_1]
\otimes \cdots \otimes \F[G_s]$. Furthermore, the 
factorization $G_i = G_{i,1} \times \cdots \times G_{i,t_i}$
implies that $\F[G_i]$ is isomorphic, as an $\F$-algebra, to
$$\F[x_{i,1},\dots,x_{i,t_i}]/\left \langle
x_{i,1}^{p_i^{b_{1}}}-1,\dots,x_{i,t_i}^{p_i^{b_{i,t_i}}}-1\right\rangle
=\B_{p_i,\boldsymbol{b}_i,\boldsymbol{x}_i},$$ with $\boldsymbol{b}_i
= (b_{i,1},\dots,b_{i,t_i})$ and $\boldsymbol{x}_i =
(x_{i,1},\dots,x_{i,t_i})$. Given $\beta$ on the power basis in
$\gamma_{1,1},\dots,\gamma_{s,t_s}$, we obtain its image $B$ in
$\B_{p_1,\boldsymbol{b}_1,\boldsymbol{x}_1} \otimes \cdots \otimes
\B_{p_s,\boldsymbol{b}_s,\boldsymbol{x}_s}$ simply by renaming
$\gamma_{i,j}$ as $x_{i,j}$, for all $i,j$.

For $i \le s$, by Lemma~\ref{lemma:alg}, there exist integers
$c_{i,1},\dots,c_{i,N_i}$ such that
$\B_{p_i,\boldsymbol{b}_i,\boldsymbol{x}_i}$ is isomorphic to an
algebra $\D_{p_i, \boldsymbol{b}_i, z_i}$, with factors 
$\F[z_i]/\langle \Phi_{{p_i}^{c_{i,j}}}(z_i) \rangle$.
By distributivity of the tensor product over direct products, we
deduce that $\B_{p_1,\boldsymbol{b}_1,\boldsymbol{x}_1} \otimes \cdots
\otimes \B_{p_s,\boldsymbol{b}_s,\boldsymbol{x}_s}$ is isomorphic to
the product of algebras
 \begin{equation}\label{eq:prod}
\text{\small $\prod$}_{\boldsymbol{j}}~ \F[z_1,\dots,z_s]/
\langle \Phi_{{p_1}^{c_{1,j_1}}}(z_1),\dots, \Phi_{{p_s}^{c_{s,j_s}}}(z_s) \rangle,   
 \end{equation}
for all indices $\boldsymbol{j}=(j_1,\dots,j_s)$, with
$j_1 =1,\dots,N_1,\dots,j_s=1,\dots,N_s$;
call $\Gamma$ the isomorphism. Given $B$ in $\B_{p_1,\boldsymbol{b}_1,\boldsymbol{x}_1} \otimes
\cdots \otimes \B_{p_s,\boldsymbol{b}_s,\boldsymbol{x}_s}$,
Lemma~\ref{lemma:alg} also implies that $B':=\Gamma(B)$ can be
computed in softly linear time $\tilde{O}(|G|)$ (apply the isomorphism
corresponding to $\boldsymbol{x}_1$ coordinate-wise with respect to
all other variables, then deal with $\boldsymbol{x}_2$, etc).
The codomain in~\eqref{eq:prod} is the product of all $\H_{\boldsymbol{p},\boldsymbol{c}_{\boldsymbol{j}},\boldsymbol{z}}$,
with 
$$\boldsymbol{p}=(p_1,\dots,p_s),\quad \boldsymbol{c}=(c_{1,j_1},\dots,c_{s,j_s}),\quad \boldsymbol{z}=(z_1,\dots,z_s).$$
Apply Lemma~\ref{lemma:distinctP} to all 
$\H_{\boldsymbol{p},\boldsymbol{c}_{\boldsymbol{j}},\boldsymbol{z}}$ to obtain
an $\F$-algebra isomorphism
$$\Gamma': \text{\small $\prod$}_{\boldsymbol{j}}~
\H_{\boldsymbol{p},\boldsymbol{c}_{\boldsymbol{j}},\boldsymbol{z}} \to
\text{\small $\prod$}_{\boldsymbol{j}} ~\F[z]/\langle
\Phi_{d_{\boldsymbol{j}}}(z) \rangle,$$ for certain integers
$d_{\boldsymbol{j}}$. The lemma implies that given $B'$,
$B'':=\Gamma'(B')$ can be computed in softly linear time
$\tilde{O}(|G|)$ as well. Invertibility of $\beta \in \F[G]$ is
equivalent to $A''$ being invertible, that is, to all its components
being invertible in the respective factors $\F[z]/\langle
\Phi_{d_{\boldsymbol{j}}}(z) \rangle$. Invertibility in such an
algebra can be tested in softly linear time by applying the fast
extended GCD algorithm~\cite[Chapter~11]{vzGathen13}, so our conclusion follows. 

\smallskip

Together with Proposition \ref{prop:abelian}, the above result proves
the first part of Theorem \ref{thm:main}.


\subsection{Metacyclic Groups}

In this last section, we study the invertibility problem for a
metacyclic group $G$. Instead of using an $\F$-algebra isomorphism, as
we did above, we will use an injective homomorphism, whose image will
be easy to compute. This is the object of the following lemma, where
the map is inspired by the one used in \cite[\S 47]{Curtis}.

Assume that $G = \langle \sigma , \tau : \sigma^m = 1, \tau^s =
\sigma^t, \tau^{-1} \sigma \tau = \sigma^r \rangle$, where $r^s = 1
\bmod m$ and $rt = t \bmod m$; in particular, $n=|G|$ is equal to
$ms$. Define $\A:=\F[z]/\langle z^m-1\rangle$ and let $\zbar$ be the
image of $z$ in $\A$.

\begin{lemma}\label{prop:metinjection}
The mapping
\[
  \begin{array}{cccc}
\psi:& \F [G]& \to& M_s(\A) \\
&\sigma &\mapsto& \mathrm{Diag}(\zbar, \zbar^r , \ldots, \zbar^{r^{s-1}})\\
&\tau &\mapsto &
\left[ \begin{array}{l|l}
0 & \zbar\\
\hline
\mat I_{s-1} & 0
\end{array}
\right],
  \end{array}
\]
an injective homomorphism of $\F$-algebras.
\end{lemma}
\begin{proof}
It is straightforward to verify that $\psi(\sigma)^m = \mat I_m$,
$\psi(\tau)^s = \psi(\sigma)^t$ and $\psi_i(\sigma) \psi_i(\tau)
=\psi(\tau) \psi_i(\sigma)^r$; this shows that $\psi$ is a well-defined $\F$-algebras homomorphism.

Take $\beta \in \F[G]$, and write it $\beta = \sum_{j = 0}^{s-1}
\left( \sum_{i = 0}^{m-1} b_{i,j} \sigma^i \right) \tau^j$. For
$j=0,\dots,s-1$, define $F_j(x) := \sum_{i = 0}^{m-1} b_{i,j}x^i \in
\F[x]$ and, for $1 \leq i,j \leq s$,
$$F_{i,j} := F_{i-1}(\zbar^{r^{j-1}}).$$
Then, $\psi(\beta)$ is the matrix
\begin{equation}\label{eq:injection}
\left[\begin{array}{llll}
F_{1,1} &  \cdots	&	\zbar F_{3,s-1} & \zbar F_{2,1}\\
F_{2,2} & F_{1,2}& \cdots & \zbar F_{3,s}\\
\vdots &\ddots & \ddots& \vdots\\
F_{s,s}& \cdots & F_{2,s}	& F_{1,s}
\end{array}
\right].
\end{equation}
If $\beta$ is in $\mathrm{Ker}(\psi)$, we get $F_i(\zbar) =0$, that is,
$F_i \bmod (z^m-1)=0$, for $0 \leq i < s$. Since all $F_i$'s have degree 
less than $m$, they are all zero.
\end{proof}

We conclude this section with two algorithms that test whether
$\psi(\beta) \in M_s(\A)$ is invertible, for a given $\beta$ in
$\F[G]$. Minor difficulties will arise as we work over $\A$, since
$\A$ is not a field, but a product of fields (if the irreducible
factorization of $z^m-1$ in $\F[z]$ is available, one would simply use
the Chinese Remainder theorem and work in field extensions of $\F$).

\begin{corollary}\label{coro:test_meta}
  Given $\beta$ in $\F[G]$, one can test if $\beta$ is a unit in
  $\F[G]$ either by a deterministic algorithm that uses
  $\tilde{O}(s^{2.7} m)$ operations in $\F$, or a Monte Carlo one that
  uses $\tilde{O}(n^2)$ operations in $\F$.
\end{corollary}
The second statement provides the last part of the proof of Theorem
\ref{thm:main}. Note that the first algorithm gives a better cost in
many cases. For instance, if $s \leq m$, the first algorithm uses
$O(n^{1.85})$ operations in $\F$. This happens if $s$ is prime, since
then the number ${(m- \gcd(m,r-1))}/{s}$ is a positive integer,
which implies $s \leq m$ (see \cite[Theorem 47.12, Corollary 47.14
]{Curtis}).

\begin{proof}[First algorithm.]
  The first algorithm uses fast linear algebra algorithms over the ring
  $\A$. Here, we start from $\beta$ written as $\beta = \sum_{j =
    0}^{s-1} \left( \sum_{i = 0}^{m-1} b_{i,j} \sigma^i \right) \tau^j
  \in \F[G]$. Then, the proof of the previous lemma shows an explicit
  formula for $\psi(\beta)$. In order to compute this matrix, we note
  that $\zbar^{r^{j-1}} = \zbar^{r^{j-1} \bmod m}$; computing this
  element and its powers requires no arithmetic operation, so that the
  coefficients of each $F_{i,j}$ are obtained in linear time $O(m)$.
  Hence the matrix $\psi(\beta)$ can be computed in time $O(s^2m)$.

  Next, we have to determine whether $\psi(\beta)$ is a unit (the
  injectivity of $\psi$ implies that this is the case if and only if
  $\beta$ itself is a unit). This amounts to computing the determinant
  of this matrix, which can be done in $\tilde{O}(s^{2.7} m)$
  operations in $\F$, using the determinant algorithm
  of~\cite[Section~6]{KaVi04}.
\end{proof}

Before giving our second algorithm, let us point out that matrix-vector
products by $\psi(\beta)$ can be done fast.

\begin{lemma}\label{lem:multpsi}
  Given $\beta$ in $\F[G]$ and $\boldsymbol{v}$ in $\A^s$, one can
  compute $\psi(\beta) \boldsymbol{v} \in \A^s$ using $\tilde{O}(s
  m^2)$ operations in $\F$.
\end{lemma}
\begin{proof}
  We use the basis of $\F[G]$ given in~\eqref{pres2}, writing $\beta =
  \sum_{i = 0}^{m-1} \left( \sum_{j = 0}^{s-1}  b_{i,j} \tau^i
  \right) \sigma^j \in \F[G]$. We rewrite this as $\beta = \sum_{i
    = 0}^{m-1} B_i(\tau) \sigma^i$, for some $B_0,\dots,B_{m-1}$ in
  $\F[z]$ of degree less than $s$.

 Given $\boldsymbol{v}$ as above, we compute all $B_i(\psi(\tau))
 \psi(\sigma)^i \boldsymbol{v}$ independently, and add them to obtain
 $\psi(\beta) \boldsymbol{v}$. Hence, let us fix an index $i$ in
 $\{0,\dots,m-1\}$.
 The vector $\psi(\sigma)^i \boldsymbol{v}$ can be obtained by
 multiplying each entry of $\boldsymbol{v}$ by a power of $\zbar$;
 this takes $\tilde{O}(sm)$ operations in $\F$. Then, since
 $\psi(\tau)$ is the matrix of multiplication by $y$ in $\A[y]/\langle
 y^s-\zbar\rangle$, $B_i(\psi(\tau))$ is the matrix of multiplication
 by $B_i(y)$ in $\A[y]/\langle y^s-\zbar\rangle$. Thus, applying this
 matrix to a vector also takes softly linear time $\tilde{O}(sm)$.

 Adding a factor of $m$ to account for indices $i=0,\dots,m-1$, we get 
 the result.
\end{proof}

\begin{proof}[Second algorithm for Corollary~\ref{coro:test_meta}.]
  The second algorithm uses \citeauthor{Wiedemann86}'s
  \citeyear{Wiedemann86} algorithm,
  and its extension by~\citeN{KaSa91}. Extra care will be needed to
  accommodate the fact that $\A$ has zero-divisors. Let
  $F_1,\dots,F_s$ be the (unknown) irreducible factors of $z^m-1$ in
  $\F[z]$ and define $\A_i := \F[z]/\langle F_i \rangle$ for
  $i=1,\dots,s$. We write $\pi_i: \A \to\A_i$ for the canonical
  projection, and extend the notation to matrices over $\A$.

  For $\beta$ in $\F[G]$, $\mat M:=\psi(\beta)$ is invertible if and only
  if all $\mat M_i := \pi_i(\mat M)$ are. We are going to use the algorithm
  of~\cite[Section~4]{KaSa91} to compute the rank of all these matrices
  (these ranks are well-defined, since all $\A_i$'s are fields).  Let
  $\mat L$ and $\mat U$ be respectively random lower triangular and upper
  triangular Toeplitz matrices over $\A$, and define
  $\mat M':=\mat L \mat M \mat U \in M_s(\A)$. Finally, let $\mat M''$ be
  $\mat M'$, to which we adjoin a bottom row and a rightmost column of
  zeros (so it has size $s+1$), let $\mat M''_i :=\pi_i(\mat M'')$ and let
  $r_i :=\text{rank}(\mat M''_i)$, $i=1,\dots,s$. Then, all $r_i$'s are
  less than $s+1$, and $\mat M$ is invertible if and only if $r_i=s$ for
  all $i$.

  The condition that $\mat M''_i$ has rank less than $s+1$ makes it
  possible to apply~\cite[Lemma~2]{KaSa91}: for generic
  $\boldsymbol{u}_i, \boldsymbol{v}_i$ in $\A_i^{s+1}$ and diagonal
  matrix $\mat X$ in $M_{s+1}(\A_i)$, the minimal polynomial of the
  sequence $(\boldsymbol{u}_i^T (\mat M''_i \mat X_i)^j
  \boldsymbol{v}_i)_{j \ge 0}$ has degree $r_i+1$.

  To compute these degrees without knowing the factorization $z^m-1 =
  F_1 \cdots F_s$, we choose random $\boldsymbol{u}, \boldsymbol{v}$
  in $\A^{s+1}$ and diagonal matrix $\mat X$ in $M_{s+1}(\A)$.  Then,
  we compute $2s$ terms in the sequence $(\gamma_j)_{j\ge 0}$, with
  $\gamma_j:=\boldsymbol{u}^T (\mat M'' \mat X)^j
  \boldsymbol{v}$. Since multiplication by $\mat L$, $\mat U$ and
  $\mat X$ all take quasi-linear time $\tilde{O}(sm)$,
  Lemma~\ref{lem:multpsi} shows that one product by $\mat M'' \mat X$
  takes $\tilde{O}(sm^2)$ operations in $\F$. Hence, all required
  terms can be obtained in $\tilde{O}(s^2m^2)=\tilde{O}(n^2)$
  operations in $\F$.

  Finally, we apply the fast Euclidean algorithm to $\sum_{j=0}^{2s-1}
  \gamma_j y^j$ and $y^{2s}$ in the ring $\A[y]$ to find the ranks
  $r_1,\dots,r_s$.  Since $\A$ is not a field, we rely on the
  algorithm of~\cite{AcCoMa03,DaMMMScXi06}. Using $\tilde{O}(sm)$
  operations in $\F$, it reveals a partial factorization of $z^m-1$ as
  $G_1 \cdots G_t$ (the factors may not be irreducible) and integers
  $\rho_j$, $j=1,\dots,t$, such that for all $i \le s$, $j\le t$, if
  $F_i$ divides $G_j$, then $r_i=\rho_j$. This allows us to determine all
  $r_i$'s, and thus decide whether $\psi(\beta)$ is singular.
\end{proof}



\bibliographystyle{plainnat}
\bibliography{NormalBasisCharZero} 

\newcommand{\Gathen}{\relax}
\begin{thebibliography}{33}
\providecommand{\natexlab}[1]{#1}
\providecommand{\url}[1]{\texttt{#1}}
\expandafter\ifx\csname urlstyle\endcsname\relax
  \providecommand{\doi}[1]{doi: #1}\else
  \providecommand{\doi}{doi: \begingroup \urlstyle{rm}\Url}\fi

\bibitem[Accettella et~al.(2003)Accettella, Corso, and Manzini]{AcCoMa03}
C.~J. Accettella, G.~M.~Del Corso, and G.~Manzini.
\newblock Inversion of two level circulant matrices over $\mathbb{Z}_p$.
\newblock \emph{Lin. Alg. Appl.}, 366:\penalty0 5 -- 23, 2003.

\bibitem[Augot and Camion(1994)]{AugCam94}
D.~Augot and P.~Camion.
\newblock A deterministic algorithm for computing a normal basis in a finite
  field.
\newblock In P.~Charpin, editor, \emph{Proc. EUROCODE'94}, Abbaye de la
  Bussi\`ere sur Ouche, France, 1994.

\bibitem[Bach and Shallit(1996)]{BacSha96}
E.~Bach and J.~Shallit.
\newblock \emph{Algorithmic Number Theory, Volume 1: Efficient Algorithms}.
\newblock MIT Press, Cambridge, MA, 1996.

\bibitem[Bostan et~al.(2006)Bostan, Flajolet, Salvy, and Schost]{BoFlSaSc06}
A.~Bostan, P.~Flajolet, B.~Salvy, and {\'E}.~Schost.
\newblock Fast computation of special resultants.
\newblock \emph{J. Symbolic Comput.}, 41\penalty0 (1):\penalty0 1--29, 2006.

\bibitem[Brent and Kung(1978)]{BrKu78}
R.~P. Brent and H.~T. Kung.
\newblock Fast algorithms for manipulating formal power series.
\newblock \emph{Journal of the Association for Computing Machinery},
  25\penalty0 (4):\penalty0 581--595, 1978.

\bibitem[Canny et~al.(1989)Canny, Kaltofen, and Lakshman]{CaKaYa89}
J.~Canny, E.~Kaltofen, and Y.~Lakshman.
\newblock Solving systems of non-linear polynomial equations faster.
\newblock In \emph{ISSAC'89}, pages 121--128. ACM, 1989.

\bibitem[Clausen and M\"{u}ller(2004)]{ClaMu04}
M.~Clausen and M.~M\"{u}ller.
\newblock Generating fast {F}ourier transforms of solvable groups.
\newblock \emph{J. Symbolic Comput.}, 37\penalty0 (2):\penalty0 137--156, 2004.
\newblock ISSN 0747-7171.
\newblock \doi{10.1016/j.jsc.2002.06.006}.
\newblock URL \url{https://doi.org/10.1016/j.jsc.2002.06.006}.

\bibitem[Curtis and Reiner(1988)]{Curtis}
C.~Curtis and I.~Reiner.
\newblock \emph{Representation theory of finite groups and associative
  algebras}.
\newblock Wiley Classics Library. John Wiley \& Sons, Inc., New York, New York,
  1988.
\newblock ISBN 0-471-60845-9.
\newblock Reprint of the 1962 original, A Wiley-Interscience Publication.

\bibitem[{Dahan} et~al.(2006){Dahan}, {{Moreno Maza}}, {Schost}, and
  {Xie}]{DaMMMScXi06}
X.~{Dahan}, M.~{{Moreno Maza}}, {\'E}.~{Schost}, and Y.~{Xie}.
\newblock On the complexity of the {{D5}} principle.
\newblock In \emph{Proc. of {\em Transgressive Computing 2006}}, Granada,
  Spain, 2006.

\bibitem[Gao et~al.(2000)Gao, \Gathen{von zur Gathen}, Panario, and
  Shoup]{GaGaPaSh00}
S.~Gao, J.~\Gathen{von zur Gathen}, D.~Panario, and V.~Shoup.
\newblock Algorithms for exponentiation in finite fields.
\newblock \emph{Journal of Symbolic Computation}, 29\penalty0 (6):\penalty0
  879--889, 2000.

\bibitem[\Gathen{von zur Gathen} and Gerhard(2013)]{vzGathen13}
J.~\Gathen{von zur Gathen} and J.~Gerhard.
\newblock \emph{Modern Computer Algebra (third edition)}.
\newblock Cambridge University Press, Cambridge, U.K., 2013.
\newblock ISBN 9781107039032.

\bibitem[\Gathen{von zur Gathen} and Giesbrecht(1990)]{GatGie90}
J.~\Gathen{von zur Gathen} and M.~Giesbrecht.
\newblock Constructing normal bases in finite fields.
\newblock \emph{J. Symbolic Comput.}, 10\penalty0 (6):\penalty0 547--570, 1990.
\newblock ISSN 0747-7171.
\newblock \doi{10.1016/S0747-7171(08)80158-7}.
\newblock URL \url{http://dx.doi.org/10.1016/S0747-7171(08)80158-7}.

\bibitem[\Gathen{von zur Gathen} and Shoup(1992)]{GaSh92}
J.~\Gathen{von zur Gathen} and V.~Shoup.
\newblock Computing {F}robenius maps and factoring polynomials.
\newblock \emph{Computational Complexity}, 2\penalty0 (3):\penalty0 187--224,
  1992.

\bibitem[Girstmair(1999)]{Girstmair}
K.~Girstmair.
\newblock An algorithm for the construction of a normal basis.
\newblock \emph{J. Number Theory}, 78\penalty0 (1):\penalty0 36--45, 1999.
\newblock ISSN 0022-314X.
\newblock \doi{10.1006/jnth.1999.2388}.
\newblock URL \url{http://dx.doi.org/10.1006/jnth.1999.2388}.

\bibitem[Jamshidpey et~al.(2018)Jamshidpey, Lemire, and Schost]{Jam18}
A.~Jamshidpey, N.~Lemire, and {\'E.}~Schost.
\newblock Algebraic construction of quasi-split algebraic tori.
\newblock \emph{ArXiv:}, 1801.09629, 2018.

\bibitem[Johnson(1976)]{Johnson}
D.~L. Johnson.
\newblock \emph{Presentations of groups}.
\newblock Cambridge University Press, Cambridge-New York-Melbourne, 1976.
\newblock London Mathematical Society Lecture Notes Series, No. 22.

\bibitem[Kaltofen and Saunders(1991)]{KaSa91}
E.~Kaltofen and D.~Saunders.
\newblock On {W}iedemann's method of solving sparse linear systems.
\newblock In \emph{AAECC-9}, volume 539 of \emph{LNCS}, pages 29--38. Springer
  Verlag, 1991.

\bibitem[Kaltofen and Shoup(1998)]{KalSho98}
E.~Kaltofen and V.~Shoup.
\newblock Subquadratic-time factoring of polynomials over finite fields.
\newblock \emph{Math. Comp.}, 67\penalty0 (223):\penalty0 1179--1197, 1998.
\newblock ISSN 0025-5718.
\newblock \doi{10.1090/S0025-5718-98-00944-2}.
\newblock URL \url{https://doi.org/10.1090/S0025-5718-98-00944-2}.

\bibitem[Kaltofen and Villard(2004)]{KaVi04}
E.~Kaltofen and G.~Villard.
\newblock On the complexity of computing determinants.
\newblock \emph{Computational Complexity}, 13\penalty0 (3-4):\penalty0 91--130,
  2004.
\newblock ISSN 1016-3328.
\newblock \doi{10.1007/s00037-004-0185-3}.
\newblock URL \url{http://dx.doi.org/10.1007/s00037-004-0185-3}.

\bibitem[Kaminski et~al.(1988)Kaminski, Kirkpatrick, and Bshouty]{KaKiBs88}
M.~Kaminski, D.G. Kirkpatrick, and N.H. Bshouty.
\newblock Addition requirements for matrix and transposed matrix products.
\newblock \emph{J. Algorithms}, 9\penalty0 (3):\penalty0 354--364, 1988.

\bibitem[Kedlaya and Umans(2011)]{KeUm11}
K.~Kedlaya and C.~Umans.
\newblock Fast polynomial factorization and modular composition.
\newblock \emph{SICOMP}, 40\penalty0 (6):\penalty0 1767--1802, 2011.

\bibitem[Lang(2002)]{Lang}
S.~Lang.
\newblock \emph{Algebra}, volume 211 of \emph{Graduate Texts in Mathematics}.
\newblock Springer-Verlag, New York, third edition, 2002.
\newblock ISBN 0-387-95385-X.
\newblock \doi{10.1007/978-1-4613-0041-0}.
\newblock URL \url{http://dx.doi.org/10.1007/978-1-4613-0041-0}.

\bibitem[Le~Gall(2014)]{LeGall14}
F.~Le~Gall.
\newblock Powers of tensors and fast matrix multiplication.
\newblock In \emph{ISSAC'14}, pages 296--303, Kobe, Japan, 2014. ACM.

\bibitem[{Le Gall} and Urrutia(2018)]{LeGall}
F.~{Le Gall} and F.~Urrutia.
\newblock Improved rectangular matrix multiplication using powers of the
  {C}oppersmith-{W}inograd tensor.
\newblock In \emph{SODA '18}, pages 1029--1046, New Orleans, USA, 2018. SIAM.

\bibitem[Lenstra(1991)]{LenstraNormal}
H.~W. Lenstra, Jr.
\newblock Finding isomorphisms between finite fields.
\newblock \emph{Math. Comp.}, 56\penalty0 (193):\penalty0 329--347, 1991.
\newblock ISSN 0025-5718.
\newblock \doi{10.2307/2008545}.
\newblock URL \url{https://doi.org/10.2307/2008545}.

\bibitem[Li et~al.(2009)Li, {Moreno Maza}, and Schost]{LiMoSc09}
X.~Li, M.~{Moreno Maza}, and {\'E}.~Schost.
\newblock Fast arithmetic for triangular sets: from theory to practice.
\newblock \emph{J. Symb. Comp.}, 44\penalty0 (7):\penalty0 891--907, 2009.

\bibitem[Lotti and Romani(1983)]{LoRo83}
G.~Lotti and F.~Romani.
\newblock On the asymptotic complexity of rectangular matrix multiplication.
\newblock \emph{Theoretical Computer Science}, 23\penalty0 (2):\penalty0
  171--185, 1983.

\bibitem[Maslen et~al.(2018)Maslen, Rockmore, and Wolff]{MaRockWol18}
D.~Maslen, D.~N. Rockmore, and S.~Wolff.
\newblock The efficient computation of {F}ourier transforms on semisimple
  algebras.
\newblock \emph{J. Fourier Anal. Appl.}, 24\penalty0 (5):\penalty0 1377--1400,
  2018.
\newblock ISSN 1069-5869.

\bibitem[Poli(1994)]{Pol94}
A.~Poli.
\newblock A deterministic construction for normal bases of abelian extensions.
\newblock \emph{Comm. Algebra}, 22\penalty0 (12):\penalty0 4751--4757, 1994.
\newblock ISSN 0092-7872.
\newblock \doi{10.1080/00927879408825099}.
\newblock URL \url{http://dx.doi.org/10.1080/00927879408825099}.

\bibitem[Schlickewei and Stepanov(1993)]{SchSte93}
H.~Schlickewei and S.~Stepanov.
\newblock Algorithms to construct normal bases of cyclic number fields.
\newblock \emph{J. Number Theory}, 44\penalty0 (1):\penalty0 30--40, 1993.
\newblock ISSN 0022-314X.
\newblock \doi{10.1006/jnth.1993.1031}.
\newblock URL \url{https://doi.org/10.1006/jnth.1993.1031}.

\bibitem[Sch\"onhage and Strassen(1971)]{ScSt71}
A.~Sch\"onhage and V.~Strassen.
\newblock {S}chnelle {M}ultiplikation gro\ss er {Z}ahlen.
\newblock \emph{Computing}, 7:\penalty0 281--292, 1971.

\bibitem[Shoup(1995)]{Shoup}
V.~Shoup.
\newblock A new polynomial factorization algorithm and its implementation.
\newblock \emph{J. Symbolic Comput.}, 20\penalty0 (4):\penalty0 363--397, 1995.
\newblock ISSN 0747-7171.
\newblock \doi{10.1006/jsco.1995.1055}.
\newblock URL \url{https://doi.org/10.1006/jsco.1995.1055}.

\bibitem[Wiedemann(1986)]{Wiedemann86}
D.~Wiedemann.
\newblock Solving sparse linear equations over finite fields.
\newblock \emph{IEEE Transactions on Information Theory}, IT-32:\penalty0
  54--62, 1986.

\end{thebibliography}

\end{document}